\documentclass[conference]{IEEEtran}
\makeatletter
\def\ps@headings{%
\def\@oddhead{\mbox{}\scriptsize\rightmark \hfil \thepage}%
\def\@evenhead{\scriptsize\thepage \hfil \leftmark\mbox{}}%
\def\@oddfoot{}%
\def\@evenfoot{}}
\makeatother
\usepackage{amsmath,amssymb,amsfonts,amsthm,mathtools,dsfont}
\usepackage{color,verbatim,enumitem,psfrag,bbm,comment}
\usepackage[pdftex]{graphicx}

\usepackage{algorithm}
\usepackage{algorithmic}
\usepackage[T1]{fontenc}
\usepackage{url}
\usepackage[colorlinks,
            linkcolor=blue,
            citecolor=blue,
            urlcolor=magenta,
            linktocpage,
            plainpages=false,
    		pdfauthor={Siddhartha Banerjee},
            pdftex
]{hyperref}
            
\usepackage{times}
\def\EE{{\mathbb{E}}}\def\PP{{\mathbb{P}}}

\def\setZ{\mathbb{Z}}

\newcommand{\fall}{\,\forall\,}

\newtheorem{theorem}{Theorem}
\newtheorem*{theorem*}{Theorem}
\newtheorem{lemma}[theorem]{Lemma}

\newtheorem*{corollary*}{Corollary}

\newcommand{\be}{\begin{eqnarray}}
\newcommand{\ee}{\end{eqnarray}}

\begin{document}

\title{Epidemic Thresholds with External Agents}

% author names and affiliations
% use a multiple column layout for up to three different
% affiliations

\author{\IEEEauthorblockN{Siddhartha Banerjee, %\IEEEauthorrefmark{1},
    Avhishek Chatterjee, %\IEEEauthorrefmark{1}, 
    Sanjay Shakkottai %\IEEEauthorrefmark{1}
    }
  \IEEEauthorblockA{
    %\IEEEauthorrefmark{1}
    Department of Electrical and Computer Engineering\\
    The University of Texas at Austin, USA \\
    Email: \{siddhartha, avhishek\}@utexas.edu,shakkott@austin.utexas.edu}
}

\maketitle
\begin{abstract}
We study the effect of external infection sources on phase transitions in epidemic processes. In particular, we consider an epidemic spreading on a network via the SIS/SIR dynamics, which in addition is aided by external agents - sources unconstrained by the graph, but possessing a limited infection rate or virulence. Such a model captures many existing models of externally aided epidemics, and finds use in many settings - epidemiology, marketing and advertising, network robustness, etc. We provide a detailed characterization of the impact of external agents on epidemic thresholds. In particular, for the SIS model, we show that any external infection strategy with constant virulence either fails to significantly affect the lifetime of an epidemic, or at best, sustains the epidemic for a lifetime which is polynomial in the number of nodes. On the other hand, a random external-infection strategy, with rate increasing linearly in the number of infected nodes, succeeds under some conditions to sustain an exponential epidemic lifetime. We obtain similar sharp thresholds for the SIR model, and discuss the relevance of our results in a variety of settings.
\end{abstract}

\section{Introduction}
\label{sec:intro}

We study epidemic processes on large graphs in the presence of external agents - i.e., settings in which information/infection spreads in a network via an \emph{intrinsic infection} process (i.e., along the edges of the graph), but in addition, is aided in its spread by \emph{external sources} unconstrained by the graph topology. More specifically, we focus on \emph{characterizing phase transition behaviour in SIS/SIR epidemic processes} in the presence external infection sources of bounded virulence, but whose position is unconstrained by the graph.

Epidemic processes are widely used as an abstraction for various real-world phenomena -- human infections, computer viruses, rumors, information broadcast, product advertising, etc. Though there are many models for such epidemic processes, at a high level they can be divided in two groups -- those in which nodes once infected remain forever in such a state, and those in which nodes can recover from the infection. Henceforth, to distinguish between the two, we refer to the former as \emph{spreading processes}, reserving the term \emph{epidemic process} for the latter. Spreading processes help understand the dynamics of one-way dissemination of some object in a network -- for example: a rumor, a software update, a fitter genetic mutation, etc. On the other hand, epidemic processes are used to study transient infection processes, i.e., those which spread for a while but eventually die out in finite networks. 

A standard model for a \emph{spreading process} is the SI (Susceptible-Infected) dynamics, wherein a `susceptible' node (S) transitions to being `infected' (I) at a rate proportional to the number of infected neighboring nodes. An important metric here is the spreading time, i.e., the time taken for all nodes to get infected. In our earlier work \cite{gopban11}, it was shown that external agents can dramatically shorten the spreading time of the SI process in some graphs - further this remains true even if the agents infect in a random manner.

On the other hand, the metric of interest in \emph{epidemic processes} is usually the \emph{extinction time} -- the time when the infection dies out in the system. A non-trivial aspect of epidemic processes is that they exhibit \emph{phase transition} phenomena - when the infection rate exceeds a threshold, then the infection abruptly switches from being short-lived to being persistent (i.e., with a very large extinction time; in case of infinite graphs, they never die out). This aspect makes them interesting from a mathematical perspective, and also is important for understanding real-life settings, wherein such threshold phenomena have been observed empirically.

The two most widely-used models for epidemic processes are the Susceptible-Infected-Susceptible (SIS) \cite{massganesh05}, and the Susceptible-Infected-Resistant (SIR) dynamics \cite{draief08}. In both models, as in the SI model, a node in the susceptible state (S) transitions to the infected state (I) at a rate equal to an \emph{infection rate} times the number of infected neighbors -- however nodes do not remain infected forever. In the SIS model, an infected node transitions back to the susceptible state (S) according to a \emph{cure rate} -- the critical parameter herein is the ratio of the infection to cure rate, denoted as $\beta$\footnote{It turns out that without loss of generality, we can normalize the cure rate to $1$ -- the ratio of the two is sufficient to characterize the system.}. In the SIR dynamics, a node in state (I) upon being cured enters a \emph{resistant state} (R), whereupon it is removed from the system. Note that the absorbing states of the SIS/SIR models correspond to the epidemic becoming extinct -- for the SIS process, this is when all nodes are in state $(S)$, while for the SIR process it is when all nodes are either in state (S) or (R)\footnote{Specifically, an absorbing state of the SIR model consists of any connected component of nodes in state (R), and the rest in state (S)}. In this work, we characterize the \emph{impact of external sources on phase transition phenomena in SIS/SIR epidemics.}.

In addition to technical contributions to the understanding of epidemic processes (in terms of new techniques and intuition for threshold phenomena), our results are significant in terms of their interpretation for certain real-world settings. Our analysis captures both the worst-case perspective (by how much can the extinction time of a `harmful' epidemic be extended via an \emph{adversarial} external agent?) and the design perspective (how can we design the external policy to prolong the lifetime of `useful' epidemic). In particular, it is instructive to %interpret our results for 
keep in mind two particular examples:
\begin{itemize}[nolistsep,noitemsep]
\item \emph{Epidemiology} - Characterizing the spread of a human disease/computer worm in a network, aided by external agents. Here the aim is to understand the worst case scenario that could be induced by external agents unconstrained by locality/geography. In human diseases, the external agents correspond to long-distance travellers; for the propagation of computer viruses spreading on a local network, long-range jumps are via the Internet, or sometimes, through portable disks.
\item \emph{Product advertising} - Enhancing brand recall amongst consumers via viral means (word-of mouth, viral ads on online social networks, etc.) coupled with long-range `broadcast' advertising (TV advertising, targeting customers for special offers, etc.). The aim here is the exact opposite of the previous case -- how can product awareness be sustained for a longer time via well-designed advertising strategies.
\end{itemize}
%In the next section, we briefly describe our model and, how it encompasses several existing models of long-range infection sources. Next, we give an overview of our results, and highlight the implications of our results for the two settings mentioned above.

\subsection{Our Contributions}
\label{ssec:contrib}

We consider graphs $G_n = (V,E)$, parametrized by the number of nodes $|V|=n$, in which an epidemic commences spreading through two interacting processes: an \emph{intrinsic spread}, and an additional \emph{external infection}. The intrinsic spread follows either the SIS or SIR dynamics (refer Section \ref{sec:intro} for an informal description, and Section \ref{sec:model} for details). For the external infection process, similar to the model in \cite{gopban11},  \emph{each} node becomes infected (if susceptible) at a different (zero or non-zero) exponential rate at each instant \emph{in addition to the intrinsic process}; thus at time $t$, the external infection can be represented as a $|V|$-dimensional vector $\bar{L}(t)$ of external infection-rates for each node. We allow the external rates to be chosen as a function of the network state and history (\emph{omniscience}) and further, potentially designed to maximize the extinction time (\emph{adversarial}) -- the only constraint on the external process is that it has \emph{bounded virulence}, i.e., the total external rate at any instant is less than some parameter $\mu$ (which can be a function of $n$). Further, we assume that the external source exists \emph{if and only if} there is at least one infected node in the network, ensuring that our process has the same absorbing states as the original epidemics. %See Section \ref{sec:model} for formal definitions.

We focus on characterizing the expected time to extinction $\EE[T_{SIS}]$ (and in case of the SIR model, the number of eventual infected nodes) as a function of four factors: $(i)$ the intrinsic infection rate $\beta$ (with intrinsic cure rate normalized to $1$), $(ii)$ the external virulence $\mu$, $(iii)$ the graph topology, and $iv.$ the external infection policy.
%Unless otherwise mentioned, we assume the infection starts in the network at a single node, arbitrarily chosen; 
Our main results are as follows:
\begin{enumerate}[nolistsep,noitemsep]
\item (\textbf{Subcritical SIS epidemic}): Under the SIS dynamics with $\mu=O(1)$ (constant external infection), for \emph{any} graph $G$ with maximum degree $d_{\max}$, and \emph{any} external infection policy, if $\beta d_{\max}<1$, then $\EE[T_{SIS}]=O(1)$, i.e., \emph{the epidemic dies out in constant time}. 

For comparison, the critical threshold for subcritical epidemics \emph{without} external infection is $\beta\lambda_1<1$ \cite{massganesh05}, where $\lambda_1$ is the leading eigenvalue of the graph adjacency matrix; note that $d_{avg}\leq\lambda_1\leq d_{\max}$.

\item (\textbf{Critical SIS epidemic}): Under the SIS dynamics with $\mu=O(1)$, for any graph if $\frac{1}{d_{\max}}\ll\beta<\frac{1}{\lambda_1}$, i.e., order-wise greater than $(d_{\max})^{-1}$, but less than $(\lambda_1)^{-1}$, then \emph{there exists an external infection policy} resulting in an extinction time which is \emph{at least polynomial in $n$}. 

Furthermore, \emph{this is tight} in that for any graph with sufficiently large maximum-degree, and any infection policy with the above choices of $\beta$ and $\mu$, the extinction time is also bounded by a \emph{polynomial in $n$}. To the best of our knowledge, this is the first rigorous demonstration of such a polynomial-lifetime regime in SIS epidemics.

\item (\textbf{Supercritical SIS epidemic}): Next, we consider the problem of designing an external infection strategy so as to result in a \emph{supercritical epidemic}, i.e., $\EE[T_{SIS}]=e^{n^{\Omega(1)}}$. Without external infection, the best-known sufficient condition for this is $\beta\eta\geq 1$ \cite{massganesh05}, where $\eta$ is the graph conductance\footnote{More specifically, the local conductance, i.e., the isoperimetric number of all sets of size $\leq n^{\alpha}$ for some $\alpha>0$.}. We show that if the external rate $\mu$ grows with the number of infected nodes $|I|$ as $\gamma|I|$, for some constant $\gamma >0$, and up to some $|I|=n^{\alpha}$, then even a \emph{random external infection strategy} ensures exponential epidemic lifetime provided $\beta\eta+\gamma\geq 1$.

\item (\textbf{Subcritical SIR epidemic}): Finally, we show that $\mu=\Theta(1)$ is a tight threshold for subcritical epidemics -- if $\beta d_{\max}< 1$, then the number of infected nodes is \emph{sub-polynomial in $n$}, whereas if $\beta$ is orderwise greater, then it is possible to infect a polynomial number of nodes independent of where the infection starts spreading. A unique feature of our analysis is that unlike most existing work on the SIR epidemic which essentially reduce the problem to studying a static process, our model necessitates an understanding of the \emph{dynamics} of the SIR model -- we achieve this via a coupling between the SIS and SIR models. 
\end{enumerate}
We formally state these results in Section \ref{sec:results}. In addition, in Section \ref{ssec:examples}, we highlight the importance of our results through detailed analyses of several important settings -- we discuss both the import of our results in real-world examples, as well as apply them to characterize epidemic lifetimes in some important classes of networks.

\subsection{Related Work}

The spread of processes over graphs/networks has been studied in various contexts (e.g., epidemiology \cite{andersonmay92:diseasesbook,Keeling05,Daley01}, sociology \cite{Rogers2003Diffusion,Vesp12}, applied probability \cite{DraiefMass,Liggett99,pastor2007evolution}). %We briefly review some of this literature most relevant to our work.

The SIS/SIR processes has been widely studied to characterize its phase transitions. The foundational work in the probability community focused on infinite regular graphs, in particular, grids and trees \cite{Liggett99}. For finite graphs, phase transitions were first characterized via empirical \cite{kepwhite91:viruses,satves02:scalefree}, and also approximate (mean-field) techniques \cite{chakrabarti2008epidemic,VanMie09}. Other works have provided a more rigorous analysis of phase transition -- in particular, Ganesh et al. \cite{massganesh05} for general graphs, and some specific families, and Berger et al. \cite{bbcs05} for the preferential attachment graph. These papers provide conditions for subcritical and supercritical regimes under the SIS model. Draief et al. \cite{draief08} undertake a similar program for the SIR model. Our work follows in the line of these results, providing a rigorous characterization for epidemics with external sources.

The role of external sources in epidemic spread has been studied under different models. In the context of human disease spread, works include \cite{Colizza+2006} (airline effects), \cite{wangetal09:immuniz} (geographic effects) and \cite{balcan2009multiscale} (multiscale effects).
This is paralleled by studies of epidemics on computer/mobile phone networks \cite{kleinberg07:wlessepi,Moore02,wang09:spreading}. Finally, in the context of social networks, Myers et al. \cite{Myers12} consider a model similar to our model of random external infection, and use Twitter data to estimate the extent of external influence.

Several studies in literature consider network intervention to modify spread properties. The authors in \cite{kempekleintar03:influence} study enhancing epidemics via designing the seed set. In \cite{chakrabarti2008epidemic}, the authors change the graph topology via additional edges; this in some sense is an optimization view of the famous `small-world' graph constructions \cite{kleinberg2002small,newman1999scaling}.  Wagner and Anantharam \cite{Wagner05} consider redistributing edge infection-rates in a line network so as to ensure long-lasting epidemics. A complementary setting to this is considered by Borgs et al. \cite{BorgsAntidote}, wherein the cure-rate can be redistributed among nodes to minimize the infection spread. All of these works only consider static (or one-shot) policies -- we instead allow our external agents to dynamically change their policies.

Finally, our earlier work \cite{gopban11} considers the same model of external infection, but for \emph{spreading} processes (in particular, the SI dynamics). However, unlike spreading processes, where external agents always speed up spreading, this work shows that the situation is much more subtle for \emph{epidemic} processes.

\section{Model}
\label{sec:model}

%% Building on \cite{gopban11}, we describe the models for both intrinisc and external infections. 

\emph{Intrinsic infection}: As mentioned before, we consider epidemic processes on an underlying graph $G(V,E)$ (with $|V|=n$). The process evolves in continuous-time $t$ following either the SIS or the SIR dynamics. 

In the SIS model, a node exists in one of the two states: susceptible $(0)$ or infected $(1)$. The state of the network at any time $t$ is given by $\mathbf{X}(t) \in \{0,1\}^n$ (with $X_i(t)$ the state of node $i$ at time $t$). Infected nodes propagate the epidemic to neighboring susceptible nodes at rate $\beta$, and revert to being susceptible at a rate $1$. Hence at a time $t$ the transition for a node $j$ in the network follows a Markov process with transitions:
\begin{align}
0 \to 1:\mbox{at rate }\beta\sum_{i:(i,j)\in E} X_i(t),\,\,\ 1 \to 0:\mbox{at rate }1. \nonumber
\end{align}

In the SIR model, a node can take $3$ states: susceptible $(0)$, infected $(1)$ or resistant $(e)$. Now an infected node, i.e., with state $(1)$, upon recovering transitions to a resistant state $(e)$ -- subsequently it plays no further part in the dynamics. Thus the transition for node $j$ follows a Markov process:
\begin{align*}
0 \to 1&:\mbox{at rate }\beta\sum_{i:(i,j) \in E,X_i(t)\neq e} X_i(t),\\ 1 \to e&:\mbox{at rate }1.
\end{align*}
Note that state $(e)$ is an absorbing state for a node.

\emph{External infection}: Our model for external infections is the same as in \cite{gopban11}.  We formally model the external infection as a time-varying vector $L(t)=\{L_i(t):i \in V\}$, where $L_i(t)\geq 0$ is the \emph{extra} rate (i.e., in addition to any intrinsic infection rate) at which the external source tries to infect node $i$ at time $t$. We define the \emph{external virulence} as $||L(t)||_1 \leq \mu$ -- one way to visualize this is that the external source's infection times follow a Poisson process with rate $\mu$, and during each infection event (say at time $t$), the external source infects a single node $j$ with probability $\frac{L_j(t)}{\mu}$. The external infection-rate vector $L(t)$ (and the virulence $\mu$) can vary with time $t$ and can depend on the state of the network $S(t)$.

Our external infection model generalizes several models for long-range infection spreading, in particular:
\begin{itemize}[nolistsep,noitemsep]
\item $\mu=0$ reduces to the underlying SIS/SIR process.
\item \emph{Static long-range edges}: $\mu$ is the number of additional edges added; $L_i(t)$ at time $t$ is the number of long-range edges incident on node $i$ that have an infected node at the other end.
\item \emph{Dynamic long-range edges}: Same as above, but the set of additional edges can be changed over time.%%- this corresponds to choosing fresh sets of long-range infection targets depending on network state.
\item \emph{Mobile agents}: One/several mobile agents spread the infection (and which can move arbitrarily over the graph). This can also model targeted advertising, giving `special offers' for individuals, etc. 
\item \emph{Broadcasting with bandwidth constraints}: An external source with \emph{bandwidth} $\mu$, which can be shared across any set of nodes. Such a model can be used for broadcast advertising (eg., TV/magazine ads), or dissemination of software updates from a central server, etc.
\end{itemize}
Note that we do not claim our results are the tightest possible under all models of infection via external sources. Some of the above have been studied before, and it is sometimes non-trivial to analyze particular models of external infection, such as additional long-range links \cite{chakrabarti2008epidemic} or agents performing random walks on the graph \cite{DraiefGanesh}. However, our model and analysis does capture many of the salient features of these models, reproducing some existing results, and more importantly, helping characterize many new settings. At a higher level, our work suggests that such models for external infection exhibit a certain dichotomy -- external infection sources either do not affect the phase transition points (if $\mu$ is insufficient), or if they do change the threshold (i.e., if $\mu$ is large enough), then they do so even if the strategy is random.

\noindent \textbf{Notation:} We denote $\mathbb{Z}_0$ to be the collection of non-negative integers, and $\mathbb{R}_0$ to be the real numbers. Further, we use the Landau notation ($O$, $\Theta$, $\Omega$) to describe the growth rate of functions.  We use $\{ \geq_{st}, \leq_{st} \}$ to denote stochastic dominance relations between random variables; specifically $Y \geq_{st} X$ implies that $\PP[Y > r] \geq \PP[X > r]$ for all $r.$

\section{Main Results and Discussion}
\label{sec:results}

We now state our main results, and discuss their implications for different graphs, and different external agent models. We present the results and discussions here, with brief proof outlines. In Section \ref{sec:analysis}, we discuss in more detail some of the technical novelty of our analysis. The complete proofs are presented in the appendix.

\subsection{Subcritical SIS Epidemic}
\label{ssec:SISconstant}

For an SIS epidemic $\mathbf{X}(t)=\{X_i(t): i \in V\}$, the metric of interest is the \emph{extinction time}, given by: 
$$T_{SIS}=\inf\{t: \mathbf{X}(t)=\mathbf{0}\}.$$ 
Note that this definition remains consistent under our model of external agents, as we define the external infection vector $\mathbf{L}(t)$ to be $\mathbf{0}$ when $\mathbf{X}(t)=\mathbf{0}$. Note also that the distribution of $T_{SIS}$ depends on the initial condition $\mathbf{X}(0)$.

Recall that the subcritical regime is one wherein the extinction time of an epidemic is small (usually $O(1)$ or $O(\log n)$). Our main result for subcritical SIS epidemics in the presence of external infection agents is as follows:

\begin{theorem}
\label{thm:UBSIS}
Given an SIS epidemic with intrinsic infection rate $\beta$ on a graph $G$ with maximum degree $d_{max}$, if the initial number of infected nodes is $O(1)$ (arbitrarily chosen) and $\beta d_{max}<1$, then for any external infection strategy with rate $\mu=O(1)$, we have:
$$\mathbb{E}[T_{SIS}]=O(1).$$
\end{theorem}

Prior work for the intrinsic SIS epidemic (i.e., without external agents) reveals that the subcritical threshold for $\beta$ is related to $\frac{1}{\lambda_1}$, where $\lambda_1$ is the largest eigenvalue (or spectral radius) of the adjacency matrix. This fact, though earlier known empirically and via mean-field approximations, was formally established by Ganesh et al. \cite{massganesh05}, who showed that if $\beta\lambda_1<1$, then $\EE[T_{SIS}]=O(\log n)$\footnote{Note that this result does not make any restrictions on the size of the initial infected set. The $\log n$ factor is then inevitable -- essentially, it is the expected time for $n$ isolated infected nodes to recover.}. Note that for any graph, $d_{avg}\leq\lambda_1\leq d_{\max}$ (where $d_{avg}$ is the average degree). Thus for $d$-regular graphs, the spectral radius is equal to the maximum degree; moreover, for any graph, $\beta d_{max} < 1$ is a sufficient condition for the sub-critical regime.

What Theorem \ref{thm:UBSIS} shows is that under constant external infection (i.e., $\mu=O(1)$), $\beta d_{max} < 1$ remains a sufficient condition for sub-critical epidemics -- as long as the initial infected set is of size $O(1)$, the epidemic has $O(1)$ extinction time, under an \emph{arbitrary choice of the initial infected nodes}, and \emph{arbitrary external infection strategies} (including omniscient and adversarial strategies).

\begin{proof}[Proof outline]
We first embed the multi-dimensional Markov chain in a suitable $1$-dimensional Markov chain -- we do this by considering the total number of infected nodes in the network and embedding it in a Markov chain via stochastic domination arguments. Similar techniques were used to show supercritical behavior in \cite{massganesh05, bbcs05} -- however the resulting Markov chains in those works have an absorbing state at $0$ (i.e., if there are no infected nodes in the network). A critical contribution of our work is to show that it is possible to further embed the process in an \emph{ergodic Markov Chain}, by adding some virtual transitions, while preserving the quantities of interest (in this case, the time to absorption). This technique allows us to easily compute \emph{closed-form bounds for the absorbing time},even under complex external-infection policies. The complete proof is given in Section \ref{ssec:SISsubcriticalproofs} in the Appendix.
\end{proof}

Comparing with the bounds for the intrinsic infection raises the following natural question -- what happens in settings where $\mu$ is still $O(1)$, but $d_{\max}^{-1}\ll\beta <\lambda_1^{-1}$ (i.e., $\beta$ is orderwise greater than $d_{\max}^{-1}$, but below the subcriticalilty threshold). We focus on this in the next section where we show that such a setting does in fact make a difference -- it leads to a regime where the resulting epidemic has a \emph{polynomial lifetime}.

\subsection{Critical SIS Epidemic}
\label{ssec:SIScritical}

Suppose now $\mu$ is still $O(1)$, but $d_{\max}^{-1}\ll\beta <\lambda_1^{-1}$ -- in other words, it is not large enough to  escape the subcritical regime without external aid, but larger than the sufficient condition in Theorem \ref{thm:UBSIS} for $O(1)$ extinction time with external aid. The question now is whether this regime is also subcritical, or if there is a fundamental shift in the epidemic lifetime in this region brought about by the external agents. This is particularly important in settings where $d_{\max}$ and $\lambda_1$ differ greatly -- for example, graphs with a \emph{power law} degree distribution.

In this section we show that there is in fact a fundamental shift in this regime -- for a large class of graphs, the presence of the external source causes an orderwise change in the epidemic lifetime, which we characterize in a tight manner. First, existing results for subcritical epidemics without external aid (in particular, Theorem $3.1$ in \cite{massganesh05}) can be bootstrapped to derive the following upper bound on the extinction time:

\begin{theorem}
\label{thm:SISlambda}
For an SIS epidemic with $\beta \lambda_1 < 1$, originating from any set of initially infected nodes, and aided by any external infection strategy with $\mu =O(1)$, we have $$\mathbb{E}[T_{SIS}]=n^{O(1)},$$ i.e., the lifetime is at most polynomial in the number of nodes.
\end{theorem}

\begin{proof}[Proof Outline]
The main intuition behind this result is that within any time interval that is $O(\log n)$, the external infection has a $n^{-\lambda}$ (for some $\lambda=\Theta(1)$) probability of failing to infect even a single node. Furthermore, existing results indicate that for any initial set of infected nodes, a subcritical epidemic has $O(\log n)$ expected extinction-time. Now we can combine these two facts to create a stochastically dominating process, wherein whenever the external agent successfully infects a node, we instead assume \emph{all} the nodes in the system are infected -- however by our previous observation, this can happen at most a polynomial number of times. The complete proof is given in Section \ref{ssec:SIScriticalproofs} in the appendix.
\end{proof}

Thus we know that the epidemic can last for a time at most polynomial in $n$ - however can this be achieved? The answer to this is provided in the following result:

\begin{theorem}
\label{thm:SISpolyLB}
For SIS epidemics on graphs with $d_{max}=\Omega(n^\alpha)$ for some $\alpha>0$, suppose $\beta=\Omega\left(n^{-\alpha(1-\epsilon)}\right)$ for some $\epsilon >0$. Then, for any initial infected set, there exists external infection policies with $\mu=O(1)$ such that, for some $\lambda=\Theta(1)$: $$\mathbb{E}[T_{SIS}]=\Omega(n^{-\lambda}).$$
\end{theorem}

The implication of this theorem is as follows: if we consider any graph with sufficiently large maximum degree (i.e., scaling polynomially with $n$), then an intrinsic infection rate of $\beta >> d_{max}^{-1}$ coupled with a constant external infection is sufficient to ensure that the epidemic lifetime is at least polynomial in $n$. Now the first eigenvalue of the adjacency matrix satisfies $\max\{d_{avg},\sqrt{d_{\max}}\}\leq\lambda_1\leq d_{\max}$ -- thus it can be order-wise less than $d_{\max}$. For example, for the star graph on $n$ nodes, we have $\lambda_1=\sqrt{n}$ while $d_{\max}=n$. The above theorem now implies that for such graphs, it is possible to choose an intrinsic infection rate such that $d_{max}^{-1}\ll\beta<\lambda_1^{-1}$, and then design an appropriate external infection strategy to ensure polynomial lifetime. Note that this is not unique to the star graph -- for many graphs, like power-law and other heavy-tailed degree distributions, it is known that the maximum degree is polynomial in number of nodes, and further, there is an order-wise difference between $\lambda_1$ and $d_{\max}$ (see also \cite{massganesh05}). The above theorem implies that constant external-infection can cause a polynomial-lifetime critical regime in such networks.

\begin{proof}[Proof outline]
The proof proceeds in three steps. First, given any initial infected set, the external source can infect the node with degree $d_{\max}$ with at least constant probability (by focusing its entire virulence at that node). Next, by coupling arguments, the process subsequently can be stochastically dominated by considering an SIS process on a single star of degree $d_{\max}$. The main technical ingredient is showing that the SIS infection on a star on $d_{\max}=n^{\alpha}$ nodes lasts for a time that is at least polynomial in $n$ \emph{if aided by an external source of at least constant virulence}; note that this is not true if the external source is absent (from existing results in \cite{massganesh05, bbcs05}). The complete proof is given in Section \ref{ssec:SIScriticalproofs} in the appendix.
%The complete proof is given in Section \ref{sec:proofs}.
\end{proof}

\subsection{Supercritical SIS Epidemic}
\label{ssec:SISsuper}

The results in the previous two sections indicates that $O(1)$ external infection is insufficient to push an epidemic from the subcritical to the supercritical regime - it can at most increase the lifetime to a polynomial function of $n$, but not to an exponential. The question we address now is what conditions on $\mu$ are sufficient to ensure exponential epidemic lifetimes -- i.e. $\EE[T_{SIS}]=e^{n^{\Omega(1)}}$. 

More specifically, we want to design an external-infection scheme that guarantees that epidemic is in the supercritical regime; ideally, this policy should work for any graph/initial infected set/intrinsic infection rate $\beta$. The results from the previous section show that we need $\mu$ to grow with $n$ to have any chance of reaching the supercritical regime. In fact, it turns out that just scaling with $n$ is not sufficient, as we prove the following refinement of Theorem \ref{thm:UBSIS}, that demonstrates that even if $\mu$ grows slowly with $n$, it can not escape sub-exponential extinction times:

\begin{theorem}
\label{thm:muScales}
For the SIS epidemic on graph $G$, if $\beta d_{max} < 1$ and the initial infected set is $O(1)$, then:
\begin{enumerate}[nolistsep,noitemsep]
\item $\mu=o\left(\frac{\log n}{\log \log n}\right)$\footnote{Recall $f(n)=o(g(n))$ means that $\lim_{n\rightarrow\infty}\frac{f(n)}{g(n)}=0$.} results in sub-polynomial lifetime, i.e., $\mathbb{E}[T_{SIS}]=n^{o(1)}$.
\item $\mu=O\left(\mbox{polylog}(n)\right)$ results in sub-exponential lifetime, specifically $\mathbb{E}[T_{SIS}]=O(e^{\mbox{polylog}(n)})$.
\end{enumerate} 
\end{theorem}

Thus, in order to achieve exponential epidemic lifetimes, it would appear that we need $\mu$ to be at least polynomial in $n$. However this is impractical for real-world settings -- for example, it implies having a very high advertising budget, or a large number of infected long-distance travelers. Further, such an assumption is somewhat trivial from a mathematical point of view, as essentially the external infection dominates any effects of the intrinsic infection.

However, it turns out we can circumvent this problem in the following way: instead of considering a fixed, large $\mu$, we \emph{allow $\mu$ to scale with the size of the infected population} (upto some maximum size). Note that from the previous theorem, it is clear that we need the maximum size to be polynomial in $n$, however the required external rate over time may be much smaller. This is the situation we consider next; in Section \ref{ssec:examples}, we discuss the relevance of this assumption in various settings.

To state our results, we first need to recap the best known conditions for the supercritical regime for the SIS epidemic without external infection. For this, we need to introduce the generalized isoperimetric constant $\eta(m)$ of a graph $G$ \cite{massganesh05}, defined as follows:
$$\eta(m)=\min_{S\subseteq V:|S|\leq m}\frac{|E(S,S^c)|}{|S|},$$
where $S$ is any subset of nodes (of size $\leq m$, with its complement $S^c=V\setminus S$), and $E(S,S^c)$ is the set of edges with one endpoint in $S$ and one in the complement. The isoperimetric constant captures the notion of a \emph{bottleneck set}: one which has the smallest number of edges exiting it, normalized by the number of nodes in the set. It is related to the notion of expansion on graphs -- in particular, the above definition is equivalent to a $(\frac{m}{n},\eta(m))$-expander \cite{bbcs05}. Further, via Cheeger's inequality, it is also related to the second largest eigenvalue $\lambda_2$ of the adjacency matrix. Finally, it is known \cite{massganesh05} that if $\beta\eta(n^{\alpha})>1$, then for any initial set of infected nodes, we have $\EE[T_{SIS}]=\Omega(e^{n^{\alpha}})$. We now show how this can be improved with external sources.

Suppose we define the set of infected nodes at time $t$ as $I(t)\subseteq V$. Now we have the following sufficient condition for the epidemic to enter the supercritical regime:
\begin{theorem}
\label{thm:LBSISgammaI}
Consider the SIS epidemic on graph $G$: given any $\alpha>0$, suppose the external infection rate $\mu$ scales linearly with the number of infected nodes $|I|$ as $\mu=\gamma|I|$ for $|I|<n^\alpha$. Now if the intrinsic infection rate satisfies $\beta\eta(|I|)+\gamma > 1$, then under an external infection strategy which infects nodes uniformly at random, we have $$\mathbb{E}[T_{SIS}] \geq e^{n^\alpha}.$$
\end{theorem}

We can unravel this result as follows -- to achieve supercriticality, we require:
\begin{itemize}[nolistsep,noitemsep]
\item The external infection-rate scale linearly with the number of infected nodes $|I|$ upto a maximum limit which is at least polynomial in $n$ (though it can be of as small an order as we desire).
\item The slope of the line $\gamma$ must obey $\gamma>1-\beta\eta(n^{\alpha})$: essentially, it must compensate for the gap between the intrinsic rate and the critical threshold.
\item The external infection need not be specifically designed -- in fact, it can be \emph{random}. More specifically, it is sufficient to try to infect each node at a rate $\frac{\mu}{n}$, irrespective of whether it is already infected or not.
\end{itemize}
We give an outline of the proof in Section \ref{ssec:mc}; the complete proof is provided in Section \ref{ssec:SISsupercriticalproofs} of the appendix.

Combining all the previous theorems, we get the following summary for the SIS epidemic with external agents: suppose we have an intrinsic infection with rate $\beta$ which is below the critical threshold, then adding a constant amount of additional external infection will either leave the extinction time unaffected (in a scaling sense, if $\beta d_{\max}<1$), or at best, if designed carefully, improve it to a polynomial in $n$ (if $\beta d_{\max}>>1$). If however, we want to push the epidemic into the supercritical regime, we can do so by making the external rate grow linearly with the number of infected nodes, at a rate equal to how much $\beta$ is below the threshold. The three regimes are summarized in Figure \ref{fig:3reg}. We interpret these results in more detail via several examples in Section \ref{ssec:examples}.
\begin{figure}[!h]
\centering \includegraphics[scale=0.4]{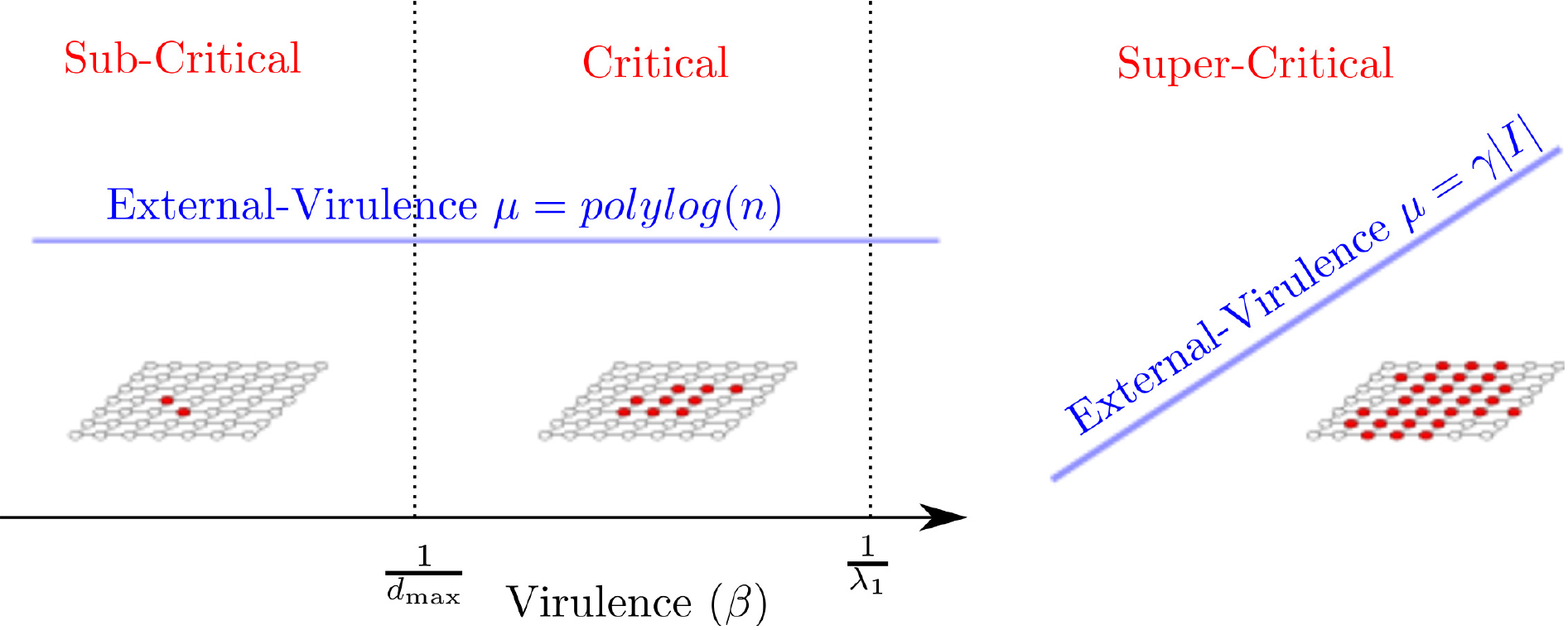}
\caption[SIS Epidemics: $3$ regimes]{Illustration of the $3$ regimes in SIS epidemics with external infection sources: For $\mu=O(\mathrm{polylog}(n))$, we get subcritical and critical regimes, depending on $\beta$; $\mu$ scaling linearly with $|I|$ results in a supercritical epidemic.}
\label{fig:3reg}
\end{figure}

\subsection{Subcritical SIR Epidemic}
\label{ssec:SIR}

Next we turn to the case of the SIR epidemic (see Section \ref{sec:model} for definitions). Recall that in this model, an absorbing state consists of a connected set of nodes in state $(e)$, and the rest in state $(0)$. We define the time to extinction $T_{SIR}$ of the epidemic as before. However a more important object of study here is the number of nodes \emph{eventually infected} by the epidemic -- we denote this as $N_{SIR}=|I(T_{SIR})|$. 

We assume throughout that the initial number of infected nodes is $O(1)$, and they are chosen arbitrarily -- this can be generalized, but since we are interested in the number eventually infected, this assumption is appropriate. We also note that there are two models which are referred to as the SIR model in literature -- the one we define in Section \ref{sec:model}, wherein the nodes transition from state $(1)$ to $(e)$ according to an $Exponential(1)$ clock, and an alternate model wherein the nodes transition from $(1)$ to $(e)$ after a \emph{deterministic time} (usually $1$ unit). We consider only the former model in this work; the results can be generalized to the latter in a straightforward way.

Sufficient conditions for subcritical SIR epidemics, i.e., with small $\EE[N_{SIR}]$, have been (rigorously) established -- in particular, Driaef et al. \cite{draief08} show that $\beta d_{\max}< 1$, then $\EE[N_{SIR}]=O(1)$\footnote{One can also get results for the case when $\beta\lambda_1<1$, but the bounds are in terms of the eigenvectors of the adjacency matrix, making them harder to interpret in general.}. Our main result in this section is a generalization of this result for the setting with external agents:

\begin{theorem}
\label{thm:SIRsubPoly}
For the SIR epidemic on graph $G$ with external infection-rate $\mu=O(1)$, if $\beta (d_{max}-1) < 1$, then for any external infection policy, we have: $$\mathbb{E}[N_{SIR}]=O(1).$$
Further, if $\mu=o\left(\frac{\log n}{\log \log n}\right)$ then the number 
of nodes eventually infected is sub-polynomial -- i.e., $\EE[N_{SIR}]=n^{o(1)}$.
\end{theorem}

This means that any infection with intrinsic rate less than $d_{max}^{-1}$ can infect only a \emph{vanishing fraction} of the nodes if the external infection strength is less than $\frac{\log n}{\log \log n}$.

\begin{proof}[Proof Outline]
As existing bounds on the SIR epidemic are obtained by observing the structure of the infection as $t \to \infty$\cite{draief08}, it is difficult to take the same approach to incorporate external infection. We instead build on our results for the SIS epidemic spread via a coupling argument. The main observation is that for SIS and SIR epidemics on the same graph, with identical intrinsic rate $\beta$, as well as external rate $\mu$ (and identical external infection strategies), the extinction time for the SIS infection \emph{stochastically dominates} that of the SIR infection. Using this, we can use various coupling arguments to get the result. The complete proof is given in Section \ref{ssec:SIRproofs} of the appendix.
\end{proof}

These conditions, in particular, for ensuring sub-polynomial number of infected nodes, are again tight in case of graphs with sufficiently large (polynomial) maximum degree. In particular we have the following result:
\begin{lemma}
\label{thm:SIRManystars}
Consider a graph $G$ with maximum degree $d_{max}=\Omega(n^\alpha)$ for some $\alpha>0$. For an SIR epidemic with $\beta=\Omega\left(n^{\epsilon-\alpha}\right)$ for $0<\epsilon<\alpha$, aided by external infection-rate $\mu=\Omega(1)$, then there exists an external infection policy such that $\mathbb{E}[N_{SIR}]=\Omega(n^k)$ for some $k>0$.
\end{lemma}

Thus constant external infection does not significantly change the number of infected nodes when the underlying epidemic is subcritical.

\subsection{Discussion and Examples}
\label{ssec:examples}

We conclude the presentation of our results by showing how they can be applied to various settings. 

\noindent\textbf{Epidemiology:} Both the SIS and SIR models are widely used to model human epidemics \cite{satves02:scalefree}, and also computer viruses/worms \cite{kepwhite91:viruses}. SIR dynamics is used in cases where recovery from an infection prevents re-infection (e.g., resistance in humans, patches for computer viruses); the SIS dynamics more closely models cases where an infection may return via mutations in the parasite (in fact this is true even for computer worms, which have been known to circumvent patches by exploiting multiple vulnerabilities). Finally, external sources have long been known to promote the spread of such epidemics, be it via long-distance air travel/shipping propagating human diseases, or worms spreading via the Internet.

What do our results imply for such settings? On the positive side, Theorems \ref{thm:UBSIS} and \ref{thm:SISlambda} show that an external source of constant virulence  can not cause an infection to become supercritical -- further, Theorem \ref{thm:SISpolyLB} reaffirms the intuition that the \emph{critical nodes that need to be defended against the infection are the ones which have high-degree}. This is in a sense complementary to the results of Borgs et al. \cite{BorgsAntidote}, who show that an effective strategy for distributing \emph{antidotes} (i.e., augmenting recovery rates) is to provide it to nodes proportional to their degree. Our result can be viewed in terms of designing preventive measures -- to control epidemic spread, one should vaccinate/quarantine nodes with a high degree.

However, it can be argued that in many cases, the number of external agents actually does scale as $\gamma|I|$ - for example in case of long-distance travel, if we assume that the people traveling are randomly chosen (irrespective of whether infected or not), then the number of infected people traveling does scale linearly as the number of infected people. In such a regime, the import of Theorem \ref{thm:LBSISgammaI} is more grim -- it suggests that an infection can become supercritical if such travel is not curbed before a certain point, and further this is not affected if the travel is random. This matches the empirical observation that epidemics which have a long incubation period (i.e., do not show symptoms in carriers till long after infection) often tend to become more widespread.

\noindent\textbf{Advertising:} It is widely accepted that the recall of a product in the minds of consumers is jointly affected by its prevalence in social circles (e.g., if many your friends use a certain product, then you would tend to use it as well) and the use of advertising. Moreover, advertising can be both via broadcast media (e.g., TV/magazine ads, billboards) and also, increasingly, viral media (e.g., word of mouth, videos on Youtube, etc.). This scenario is very well modeled by the SIS/SIR epidemic -- the epidemic represents the recall of a brand, which is strengthened by local interactions, and also by broadcast advertising; further, the rate of TV/magazine ads are subject to advertising budgets.

In this domain, our results make very clear recommendations -- for any brand, a constant advertising budget is not sufficient to ensure long-lasting recall. However, increasing the advertising budget proportional to the number of customers (which is viable since the revenues also increase) can achieve supercritical behavior. In terms of advertising strategy, our results suggest that if the budget is small, then targeted advertising, biased towards highly-networked consumers, is more beneficial - however, if the budget is large enough (and scaling linearly with the number of customers), then broadcast advertising (which can be thought of as a random infection strategy) is sufficient.

\noindent\textbf{Popular Models for Network Formation:} Our results can be applied to specific graph families to obtain thresholds for the different regimes -- however, they may not be the tightest in all graph families. Such a program is carried out in \cite{massganesh05}, who characterize $\lambda_1$ and $\eta(m)$ for several families of graphs, including random graph families such as the $G(n,p)$ and certain families of power-law graphs -- since our thresholds are also in terms of the same quantities, these results carry over to our setting, and in addition we get additional results for the case of external infection; for example, for graphs with sufficiently high maximum degree, we show that $O(1)$ external infection rate is sufficient to push the lifetime to $poly(n)$.

However, our results allow for analysis of more sophisticated models of external infection. To illustrate this, consider the example of Kleinberg's small-world network construction \cite{kleinberg2002small} -- this consists of a $2$-dimensional grid, where each node has an additional long range link, the other end of which is selected according to some given distribution. Viewing these long-range links as external-infection sources, we see that the external infection rate increases roughly as $\beta|I|$ (in particular, if the long-range link is chosen uniformly, then $\mu=\beta\left(|I|-n^{-1+\alpha}\right)$ as long as $|I|\leq n^{\alpha}$). Now, using Theorem \ref{thm:LBSISgammaI}, we have that if $\beta>1$, then the epidemic is supercritical.

\section{Proofs of Selected Results}
\label{sec:analysis}

Due to lack of space, we defer complete proofs to the appendix. Instead, we outline some of the main technical novelty behind our results; these ideas may be of independent interest for tackling other questions regarding epidemic processes.

\subsection{Embedding in a Dominating Ergodic Markov Chain}
\label{ssec:mc}

The main idea for the results in the subcritical and supercritical regimes is to study the dynamics of the \emph{total number of infected nodes} in the network at time $t$. This depends on the exact network topology; however it is possible to stochastically dominate it via appropriate $1-$dimensional Markov Chains. 

\begin{figure}[!h]
\centering \includegraphics[scale=0.65]{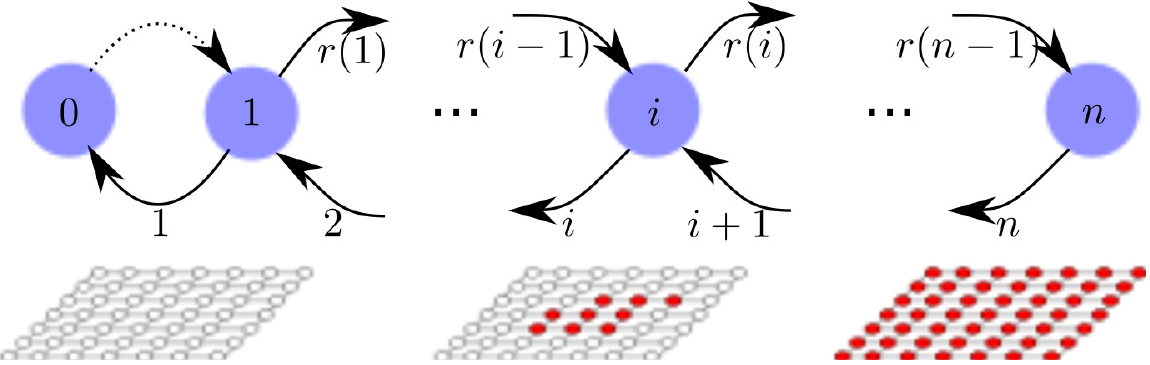}
\caption[SIS Epidemic: Dominating Aggregate-Chain]{Dominating chain for the total number of infected nodes. The transition rates $\{r(i)\}$ are appropriately chosen to over/under dominate the underlying process. A `virtual' $0\rightarrow 1$ transition makes the chain ergodic.}
\label{fig:1dmc}
\end{figure}

Similar techniques have been used before \cite{massganesh05,bbcs05}; however in all prior work, state $0$ is an absorbing state, and the absorption time is estimated via Martingale techniques. A technical novelty in our proof is that we embed the chain in an ergodic Markov chain, in a way that preserves the absorption time. This allows us to obtain closed-form expressions for the absorption time, and also consider more complex external infection strategies. We now outline this technique.

First, consider the continuous-time Markov chain for the aggregate infected nodes -- note that transitions from any state (i.e., increase/decrease in the number of infected nodes) happen at a rate at least $1$ and at most $\beta n^2+n+\mu\sim poly(n)$. This allows us to consider a uniformized discrete-time Markov chain; from the previous observation, we see that the number of time-slots till absorption for the discrete chain is within a $poly(n)$ factor of the time to absorption for the continuous chain. This is the reason this technique is most useful for the subcritical/supercritical regimes, but not fine-grained enough to capture the intermediate critical regime.

Now consider a DTMC $U_k, k\in\setZ_0$ on $\{0,1,\cdots,n\}$ with state-transitions $p_{i,i+1},p_{i,i-1}>0$ for $1\leq i\leq n-1$, $p_{n,n-1}=1$, and remaining transitions having probability $0$ -- i.e., a birth-death chain, with absorbing state at $0$ and reflection at $n$. Next, consider another DTMC, denoted $\tilde{U}_k$, on the same state space with the same transition matrix as $U_k$, but in addition, with $p_{0,1}=1$ (i.e., a reflection at $0$ instead of absorption). Note that $U_k$ has an absorbing state and hence is not an ergodic DTMC; $\tilde{U}_k$, however, is an irreducible, aperiodic, finite state Markov chain, and hence ergodic \cite{bremaud:mcbook}.

The chain $\tilde{U}_k$ is often referred to as the Ehrenfest Urn Process \cite{bremaud:mcbook}; we denote $\tilde{\pi}$ as the stationary distribution of $\tilde{U}_k$. Now let $T_{1,0}$ be the time to hit $0$ for the DTMC $U$ starting from state $1$ (i.e., the time to absorption with initial state $1$). Then we have the following lemma:
\begin{lemma}
\label{lem:ehrenfest}
$$\mathbb{E}[T_{1,0}]+1=\frac{1}{\tilde{\pi}(0)}$$
More specifically, in terms of the transition probabilities of $U_k$, with $p_{0,1}$ defined to be $1$, we have:
$$\mathbb{E}[T_{1,0}] = \sum_{k=1}^n \prod_{i=1}^k \frac{p_{i-1,i}}{p_{i,i-1}}$$
\end{lemma}

\begin{proof}
First note that the distributions of $T_{1,0}$ and $\tilde{T}_{1,0}$ are same -- this is because a sample path from
state $1$ to state $0$ for $\tilde{U}_k$ does not experience a $0$ to $1$ transition. Further, since the transition from $0$ to $1$ happens with probability $1$, $\tilde{T}_{0,0}=1+\tilde{T}_{1,0}$ a.s,. and as $\tilde{U}_k$ is an ergodic Markov chain, the expectations for both the random variables exist, and also the expected return time
for any state $i$ is $\tilde{\pi}(i)^{-1}$. Putting these together we get the first claim. For the second, observe that for a birth-death DTMC with transition provabilities $p_{i,i+1}$ and $p_{i,i-1}$, we can exactly evaluate the steady state distributions using detailed balance equations -- this gives us the closed-form expression for $\mathbb{E}[T_{1,0}]$.
\end{proof}

The above lemma can be extended to a setting where the initial number of infected nodes is some $L>0$. As before, we construct a Markov chain $U^L_k$ that has the same
transition probabilities as $U_k$, but now with an additional transition $p_{0,k}=1$ -- again we have $U^L_k$ is irreducible, finite state, and hence ergodic. Let $\pi^L$ denote the stationary distribution for $U^L_k$. We have the following lemma:
\begin{lemma}
\label{lem:kEhrenfest}
%$$\mathbb{E}[T_{K,0}]+1=\frac{1}{\pi^K(0)}$$
%In terms of the transition probabilities:
$$\mathbb{E}[T_{L,0}]=\sum_{k=1}^Lf_k + f_L \sum_{k=L+1}^n\prod_{i=L+1}^k \frac{p_{i-1,i}}{p_{i,i-1}}$$
where $f_k$s are quantities that depend only on the transition probabilities for states in $\{0,1,\cdots,L\}$.
\end{lemma}

\begin{proof}
The proof is identical to that of Lemma \ref{lem:ehrenfest} -- the only difference is that the chain is a birth-death chain for states $\geq L$. Thus the expression for the stationary probability of state $0$ has additional terms which depend on transition-rates of states between $1$ and $L$; in our applications, the terms $f_k,f_L$ are all $O(1)$, and do not affect the scaling.
\end{proof}

The above lemmas allow us to get detailed characterizations for the subcritical and supercritical regimes. For complete proofs of Theorems \ref{thm:UBSIS}, \ref{thm:muScales} and \ref{thm:LBSISgammaI}, refer Section \ref{ssec:SISsupercriticalproofs} in the Appendix.

\subsection{Dynamic Analysis of the SIR Epidemic}
\label{ssec:sirpf}

Studying thresholds for the SIR model with external infection requires understanding the transient behavior of the process. To this end, we develop the following coupling argument relating the dynamics of the SIR to the SIS model:

\begin{lemma}
\label{lem:SISubSIRtime}
Consider an SIR epidemic on graph $G$ with intrinsic infection-rate $\beta$ and a given external infection policy with infection-rate $\mu$. For an SIS epidemic on the same graph with identical external infection policy (i.e., same $L(t)\fall t$), we have:
$$T_{SIS} \geq_{st} T_{SIR}.$$
\end{lemma}

\begin{proof}
The result follows from a simple coupling argument, wherein when a node gets recovered in the SIR model, it becomes susceptible in the SIS model. Now, given that the external infection policies are identical, the SIS epidemic dominates the SIR epidemic with respect to extinction time. 
\end{proof}

\bibliographystyle{IEEEtran}
\bibliography{reference}

% Generated by IEEEtran.bst, version: 1.13 (2008/09/30)
\begin{thebibliography}{10}
\providecommand{\url}[1]{#1}
\csname url@samestyle\endcsname
\providecommand{\newblock}{\relax}
\providecommand{\bibinfo}[2]{#2}
\providecommand{\BIBentrySTDinterwordspacing}{\spaceskip=0pt\relax}
\providecommand{\BIBentryALTinterwordstretchfactor}{4}
\providecommand{\BIBentryALTinterwordspacing}{\spaceskip=\fontdimen2\font plus
\BIBentryALTinterwordstretchfactor\fontdimen3\font minus
  \fontdimen4\font\relax}
\providecommand{\BIBforeignlanguage}[2]{{%
\expandafter\ifx\csname l@#1\endcsname\relax
\typeout{** WARNING: IEEEtran.bst: No hyphenation pattern has been}%
\typeout{** loaded for the language `#1'. Using the pattern for}%
\typeout{** the default language instead.}%
\else
\language=\csname l@#1\endcsname
\fi
#2}}
\providecommand{\BIBdecl}{\relax}
\BIBdecl

\bibitem{gopban11}
A.~Gopalan, S.~Banerjee, A.~K. Das, and S.~Shakkottai, ``Random mobility and
  the spread of infection,'' in \emph{IEEE INFOCOM}, 2011, pp. 999--1007, see
  also longer version titled "Epidemic Spreading with External Agents"
  (available on arXiv.org).

\bibitem{massganesh05}
A.~J. Ganesh, L.~Massouli{\'e}, and D.~F. Towsley, ``The effect of network
  topology on the spread of epidemics,'' in \emph{IEEE INFOCOM}, 2005, pp.
  1455--1466.

\bibitem{draief08}
M.~Draief, A.~Ganesh, and L.~Massoulie, ``Thresholds for virus spread on
  networks,'' \emph{The Annals of Appl. Prob.}, vol.~18, pp. 359--378, 2008.

\bibitem{andersonmay92:diseasesbook}
R.~M. Anderson and R.~M. May, \emph{Infectious Diseases of Humans Dynamics and
  Control}.\hskip 1em plus 0.5em minus 0.4em\relax Oxford University Press,
  1992.

\bibitem{Keeling05}
M.~J. Keeling and K.~T. Eames, ``Networks and epidemic models,'' \emph{Journal
  of the Royal Society Interface}, vol.~2, no.~4, pp. 295--307, 2005.

\bibitem{Daley01}
D.~J. Daley and J.~Gani, \emph{Epidemic modelling: an introduction}.\hskip 1em
  plus 0.5em minus 0.4em\relax Cambridge University Press, 2001, vol.~15.

\bibitem{Rogers2003Diffusion}
E.~M. Rogers, \emph{Diffusion of Innovations, 5th Edition}.\hskip 1em plus
  0.5em minus 0.4em\relax Free Press, August 2003.

\bibitem{Vesp12}
A.~Vespignani, ``Modelling dynamical processes in complex socio-technical
  systems,'' \emph{Nature Physics}, vol.~8, pp. 32--39, 2012.

\bibitem{DraiefMass}
M.~Draief and L.~Massouli{\'e}, \emph{Epidemics and rumours in complex
  networks, volume 369 of London Mathematical Society Lecture Notes}.\hskip 1em
  plus 0.5em minus 0.4em\relax Cambridge University Press, Cambridge, 2010.

\bibitem{Liggett99}
T.~M. Liggett, \emph{Stochastic interacting systems: contact, voter and
  exclusion processes}.\hskip 1em plus 0.5em minus 0.4em\relax Springer, 1999,
  vol. 324.

\bibitem{pastor2007evolution}
R.~Pastor-Satorras and A.~Vespignani, \emph{Evolution and structure of the
  Internet: A statistical physics approach}.\hskip 1em plus 0.5em minus
  0.4em\relax Cambridge Univ. Press, 2007.

\bibitem{kepwhite91:viruses}
J.~O. Kephart and S.~R. White, ``Directed-graph epidemiological models of
  computer viruses,'' in \emph{IEEE Symposium on Security and Privacy}, 1991,
  pp. 343--361.

\bibitem{satves02:scalefree}
R.~Pastor-Satorras and A.~Vespignani, ``{Epidemic Dynamics in Finite Size
  Scale-Free Networks},'' \emph{Phys. Rev. E}, vol.~65, no.~3, p. 035108, Mar
  2002.

\bibitem{chakrabarti2008epidemic}
D.~Chakrabarti, Y.~Wang, C.~Wang, J.~Leskovec, and C.~Faloutsos, ``Epidemic
  thresholds in real networks,'' \emph{ACM Transactions on Information and
  System Security (TISSEC)}, vol.~10, no.~4, p.~1, 2008.

\bibitem{VanMie09}
P.~Van~Mieghem, J.~Omic, and R.~Kooij, ``Virus spread in networks,''
  \emph{Networking, IEEE/ACM Transactions on}, vol.~17, no.~1, pp. 1--14, 2009.

\bibitem{bbcs05}
N.~Berger, C.~Borgs, J.~T. Chayes, and A.~Saberi, ``On the spread of viruses on
  the internet,'' in \emph{Proceedings of the sixteenth annual ACM-SIAM
  symposium on Discrete algorithms}.\hskip 1em plus 0.5em minus 0.4em\relax
  Society for Industrial and Applied Mathematics, 2005, pp. 301--310.

\bibitem{Colizza+2006}
V.~Colizza, A.~Barrat, M.~Barth\'{e}lemy, and A.~Vespignani, ``The role of the
  airline transportation network in the prediction and predictability of global
  epidemics,'' \emph{Proceedings of the National Academy of Sciences}, vol.
  103, no.~7, pp. 2015--2020, 2006.

\bibitem{wangetal09:immuniz}
B.~Wang, K.~Aihara, and B.~J. Kim, ``Immunization of geographical networks,''
  in \emph{Complex (2)}, 2009, pp. 2388--2395.

\bibitem{balcan2009multiscale}
D.~Balcan, V.~Colizza, B.~Gon{\c{c}}alves, H.~Hu, J.~J. Ramasco, and
  A.~Vespignani, ``Multiscale mobility networks and the spatial spreading of
  infectious diseases,'' \emph{Proceedings of the National Academy of
  Sciences}, vol. 106, no.~51, pp. 21\,484--21\,489, 2009.

\bibitem{kleinberg07:wlessepi}
J.~Kleinberg, ``The {W}ireless {E}pidemic,'' \emph{Nature}, vol. 449, pp.
  287--288, September 2007.

\bibitem{Moore02}
D.~Moore, C.~Shannon \emph{et~al.}, ``Code-red: a case study on the spread and
  victims of an internet worm,'' in \emph{Proceedings of the 2nd ACM SIGCOMM
  Workshop on Internet measurment}.\hskip 1em plus 0.5em minus 0.4em\relax ACM,
  2002, pp. 273--284.

\bibitem{wang09:spreading}
P.~Wang, M.~C. Gonz\'{a}lez, C.~A. Hidalgo, and A.-L. Barabasi, ``Understanding
  the spreading patterns of mobile phone viruses,'' \emph{Science}, vol. 324,
  no. 5930, pp. 1071--1076, May 2009.

\bibitem{Myers12}
S.~A. Myers, C.~Zhu, and J.~Leskovec, ``Information diffusion and external
  influence in networks,'' in \emph{Proceedings of the 18th ACM SIGKDD
  international conference on Knowledge discovery and data mining}.\hskip 1em
  plus 0.5em minus 0.4em\relax ACM, 2012, pp. 33--41.

\bibitem{kempekleintar03:influence}
D.~Kempe, J.~Kleinberg, and E.~Tardos, ``Maximizing the spread of influence
  through a social network,'' in \emph{KDD '03: Proc. 9th ACM SIGKDD
  International Conference on Knowledge Discovery and Data mining}.\hskip 1em
  plus 0.5em minus 0.4em\relax ACM, 2003, pp. 137--146.

\bibitem{kleinberg2002small}
J.~Kleinberg, ``Small-world phenomena and the dynamics of information,''
  \emph{Advances in neural information processing systems}, vol.~1, pp.
  431--438, 2002.

\bibitem{newman1999scaling}
M.~E. Newman and D.~J. Watts, ``Scaling and percolation in the small-world
  network model,'' \emph{Physical Review E}, vol.~60, no.~6, p. 7332, 1999.

\bibitem{Wagner05}
A.~B. Wagner and V.~Anantharam, ``Designing a contact process: the
  piecewise-homogeneous process on a finite set with applications,''
  \emph{Stoch. processes and their applications}, vol. 115, no.~1, pp.
  117--153, 2005.

\bibitem{BorgsAntidote}
C.~Borgs, J.~Chayes, A.~Ganesh, and A.~Saberi, ``How to distribute antidote to
  control epidemics,'' \emph{Random Structures \& Algorithms}, vol.~37, no.~2,
  pp. 204--222, 2010.

\bibitem{DraiefGanesh}
M.~Draief and A.~Ganesh, ``A random walk model for infection on graphs: spread
  of epidemics \& rumours with mobile agents,'' \emph{Discrete Event Dynamic
  Systems}, vol.~21, no.~1, pp. 41--61, 2011.

\bibitem{bremaud:mcbook}
P.~Bremaud, \emph{Markov Chains: Gibbs Fields, Monte Carlo Simulation, and
  Queues}, corrected~ed.\hskip 1em plus 0.5em minus 0.4em\relax Springer-Verlag
  New York Inc., February 2001.

\end{thebibliography}

%\pagebreak
\appendix
\section{Proofs}
\label{sec:proofs}

\subsection{Subcritical SIS Epidemics with External Agents}
\label{ssec:SISsubcriticalproofs}

\noindent We now prove Theorems \ref{thm:UBSIS} and \ref{thm:muScales}, using the techniques discussed in Section \ref{ssec:mc}.

\begin{proof}[Proof of Theorem \ref{thm:UBSIS}]
As transitions in $\mathbf{X}(t)$ are via exponential clocks, hence the system is a continuous time Markov chain -- further $T_{SIS}$, the time when $\mathbf{X}(t)$ hits $\mathbf{0}$, is a stopping time with  respect to this Markov chain. Our aim is to upper-bound the time $\mathbb{E}[T_{SIS}]$. 

To do so, we study the one-dimensional chain induced by the number of infected nodes in the network -- more specifically, we stochastically dominate the total number of infected nodes as follows: Consider a Markov chain $Z(t)\in\{0,1,2,\ldots\}$ with $Z(0)=|\mathbf{X}(0)|$ and the state transition rates being the following:
\begin{align}
& \ i \to i+1 \ \mbox{at a rate} \ \beta d_{max} i + \mu \ \mbox{for} \ i>0\nonumber \\
& \ i \to i-1 \ \mbox{at a rate} \ i \ \mbox{for} \ i>0 \nonumber
\end{align}
We claim that $Z$ stochastically dominates $|\mathbf{X}(t)$, and consequently, the time to hit $0$ for $Z$ is a.s. larger than $T_{SIS}$. This can be formally argued via a coupling argument, but essentially it arises from the fact that for any state $\mathbf{X}(t)$, the maximum rate at which the infection can grow (i.e at which $|\mathbf{X}(t)|$ can increase) is $\beta|\mathbf{X}(t)|d_{max}+\mu$ -- this is because a node can try to infect at most $d_{max}$ nodes and the external infection can increase the number of infected nodes at a maximum rate of $\mu$. Also note that for the original process $|\mathbf{X}(t)|$ is the rate at which infection decreases, since each node recovers at rate $1$. Thus whenever $Z(t)=|\mathbf{X}(t)|=i$, the rate of growth for $Z$ is higher than that for $|\mathbf{X}|$, while the rate of decrease are the same for both the processes. Hence, we can instead focus on finding an upper bound on time to hit zero for the process $Z$.

Note that $Z$ is a process on a finite state space $[n]$ and is irreducible, with an absorbing state at $0$ -- hence the expected time to hit state $0$ is finite. We next consider $Z^D$, the  embedded discrete-time chain (i.e., jump chain) of $Z$, whose transitions ($\fall i>0$) are given by: 
\begin{align*}
p_{i,i+1}=\frac{\beta d_{max}i+\mu}{(\beta d_{max}+1)i+\mu},\,\, p_{i,i-1}=\frac{i}{(\beta d_{max}+1)i+\mu}
\end{align*}
Note that as the rate of jumps for any state $i>0$ is lower bounded by $1$ -- hence the expected time spent in any state is $O(1)$, and the time to absorption in the discrete time chain is orderwise equivalent to the absorption time in the original chain. Now we can use Lemma \ref{lem:kEhrenfest} to compute the expected hitting time to zero for $Z^D$, given initial state is $1$ (or some fixed $K$). Note that for any fixed $K>0$, the transition probabilities of $Z^D$ do not scale with $n$ for states in $[K]$; hence for the stated result, it is sufficient to show that the contribution due to the states $i\in\{K,K+1,\ldots,n\}$ is $O(1)$. We have:
\begin{align*}
\sum_{k=K}^n\prod_{i=1}^k \frac{p_{i,i+1}}{p_{i,i-1}} & \leq \sum_{k=1}^n \prod_{i=1}^k \frac{\beta d_{max} i + \mu}{i}\\
%& \leq \sum_{k=1}^n \prod_{i=1}^k (\beta d_{max})^k \exp\left(\frac{\mu}{i \beta d_{max}}\right)\\
&\leq \sum_{k=1}^n (\beta d_{max})^k \exp\left(\sum_{i=1}^k \frac{\mu}{i \beta d_{max}}\right)\\
%& \leq \sum_{k=1}^n (\beta d_{max})^k \exp\left(\frac{\mu}{\beta d_{max}}\log k\right)\\
& \leq \sum_{k=1}^{\infty} (\beta d_{max})^k (k+1)^{\left(\frac{\mu}{\beta d_{max}}\right)}.
\end{align*}
Now note that for $\beta d_{max}<1$ and $\mu=O(1)$, the RHS in the last expression can be upper bound by a constant order moment of a geometric random variable -- hence it is $O(1)$. This completes the proof.
\end{proof}

Thus we have that $\beta d_{\max}<1$ and $\mu=O(1)$ is a sufficient condition for subcritical epidemics. We can also now generalize the above argument to prove Theorem \ref{thm:muScales}, wherein we bound the extinction time for $\mu$ scaling with $n$.

\begin{proof}[Proof of Theorem \ref{thm:muScales}]
We define $\gamma=\frac{1}{\beta d_{\max}}>1$ -- further we can always assume $\gamma=O(1)$, as otherwise, we can increase $\beta$, which can only increase the extinction time. Now following identical arguments as in those for Theorem \ref{thm:UBSIS}, we get that 
\begin{align*}
\EE[T_{SIS}] &\leq \sum_{k=1}^{n} (\beta d_{max})^k (k+1)^{\left(\frac{\mu}{\beta d_{max}}\right)}\leq\gamma\sum_{k=2}^{n}\gamma^{-k}k^{\gamma\mu}\\
&\leq\gamma\int_1^{\infty}\gamma^{-x}x^{\gamma\mu}dx\leq\frac{\gamma\Gamma(\gamma\mu+1)}{(\log\gamma)^{\gamma\mu+1}},
\end{align*}
where $\Gamma(x+1)=x!$. Now using Stirling's formula, and the fact that $\gamma=\Theta(1)$, we get that $\EE[T_{SIS}]=O(\mu^{2\mu})$. Now suppose $\mu=\frac{f(n)\log n}{\log\log n}$, where $f(n)$ is some function such that $\lim_{n\rightarrow\infty}f(n)=0$ (i.e., $\mu=o\left(\frac{\log n}{\log\log n}\right)$, and $f(n)=o(1)$) -- substituting, we have $E[T_{SIS}]=O(n^{f(n)})=n^{o(1)}$ -- similarly we can show that $\mu=\mbox{polylog}(n)$ is not sufficient for an exponential lifetime.
\end{proof} 

\subsection{Critical SIS Epidemics with External Agents}
\label{ssec:SIScriticalproofs}

Next we turn to Theorem \ref{thm:SISlambda}, wherein we derive a polynomial upper bound for the extinction time in the case when $\beta\lambda_1<1$: 

\begin{proof}[Proof of Theorem \ref{thm:SISlambda}]
From existing results (such as Theorem $3.1$ of Ganesh et al. \cite{massganesh05}), we have that in the absence of external agents, the time to extinction $T_{SIS}^i$ for the \emph{intrinsic} SIS process obeys $\mathbb{E}[T_{SIS}^i]=O(\log n)$ for any initial set of infected nodes, if $\beta\lambda_1<1$. We now bootstrap this to obtain our result.

Note that the external sources, under any strategy, can infect at most at a rate $\mu$ -- thus to upper-bound $\mathbb{E}[T_{SIS}]$ (the extinction time \emph{with external sources}), we can stochastically dominate it as follows: at time $0$, assume all nodes are infected, and start an exponential clock with rate $\mu$. Clearly the underlying SIS infection expires if all nodes recover before this clock expires; if however the clock stops before the SIS process dies, we restart the process with \emph{all nodes infected}, and start another exponential clock with $\mu$. By standard coupling arguments, the time to extinction for this new process stochastically dominates $T_{SIS}$ for the original process.

Given an SIS process starting with all nodes initially infected, the time to extinction without external sources is $T_{SIS}^i$; further, let $Z_{\mu}\sim Exponential(\mu)$. Then we have:
\begin{align*}
\PP[Z_{\mu}>T_{SIS}^i]&=\mathbb{E}[\exp(-\mu T_{SIS}^i)]\\
&\geq \exp(-\mu \mathbb{E}[T_{SIS}]) \geq n^{- C},
\end{align*}
where $C>0$ is some constant (independent of $n$); the second inequality above follows from Jensen's inequality, and also we have used the fact that $\mathbb{E}[T_{SIS}^i]=O(\log n)$ 
and $\mu=O(1)$.

Note that the above quantity lower-bounds the probability that, in the constructed process, the infection process expires before the clock (corresponding to the external source) stops. Hence the expected number of regenerations that take place before the process expires is upper-bounded by
$O(n^C)$. Further, the expected time between each regeneration is $\mathbb{E}[\min\{T_{SIS}^i,Z_{\mu}\}]$ -- this
is upper-bounded by $\mathbb{E}[Z_{\mu}]=\frac{1}{\mu}=O(1)$. Thus $\mathbb{E}[T_{SIS}]=O(n^{O(1)})$.
\end{proof}

To complete our analysis of the critical regime, we show that polynomial extinction-times are in fact achievable for any graphs with sufficiently large (i.e., polynomial) degree. We consider the following strategy: the external agent targets any susceptible node with degree $\Omega(n^{\alpha-\epsilon+\delta})$ (henceforth referred to as a high-degree node); else it infects a susceptible node chosen at random (in fact, for ease of exposition, we assume in the proof that the external node only infects the maximum-degree node; however it is easy to see that the above strategy is sufficient to show polynomial lifetime).

A crucial result we use in our proof is as follows: note that as $\beta$ scales as $n^{\epsilon-\alpha}$, the expected number of infection attempts due to an infected high-degree node before it recovers is $\Omega(n^\epsilon)$. Suppose now we have that a high-degree node is infected, but none of its neighbors is infected -- we show that with high probability (i.e., at least $1-\frac{1}{\mbox{poly}(n)}$), the number of neighbors in infected state at the time the \emph{first} node recovery occurs is $\mbox{poly}(n)$. Subsequently, we show that given a $\mbox{poly}(n)$ number of infected neighbors, a node gets infected with a probability at least $1-n^{-\Omega(1)}$ \emph{in presence of $O(1)$ external infection} -- note that this is not true in the absence of the external agents (when $\beta\lambda_1<1$). Combining these observations, we show that the expected time to extinction of the epidemic is $\mbox{poly}(n)$.
\begin{lemma}
\label{lem:star}
Given a star-graph on $m$ nodes, and an SIS infection with $\beta=\Omega(m^{\epsilon-1})$, for some constant $\epsilon>0$. Suppose only the hub is infected at time $t=0$ -- then at the time when the first node recovery occurs, the number of infected nodes in the system is $\Omega(m^{\epsilon /3})$ with probability at least $1-O\left(\frac{1}{m^{\epsilon /3}}\right)$.
\end{lemma}

\begin{proof}
First, note that decreasing the value of $\beta$ can only make it less likely to have several infections before a recovery -- hence we can assume that $\epsilon <1$. Define $k=m^{\epsilon /3}$. Now let $N_t$ be the number of infected nodes at time $t$ (thus $N_0=1$), and $T$ be the first time that some infected node recovers. For some constants $c_1,c_2>0$, we want to show that $\PP[N_T\geq c_1 k]$ with probability $1-c_2k^{-1}$. 

For the system at a time $t<T$, suppose $N_t=i$ -- note that the hub is still infected, since no recoveries have occurred. Then there are $2$ possible events -- either $N_t$ goes to state $i+1$ at rate $\beta(m-i+1)$, or it stops (due to node recovery) at rate $i$. Now we have:
\begin{align*}
\PP[N_T\geq k]&\geq \PP[N_t\mbox{ reaches $k$ before stopping}]\\
&=\prod_{i=1}^k\left(\frac{\beta(m-i+1)}{\beta(m-i+1)+i}\right)\\
&\geq\left(1+\frac{k}{\beta(m-k+1)}\right)^{-k}\geq 1-\frac{k^2}{\beta(m-k+1)}\\
&=1-\frac{m^{2\epsilon /3}}{m^{\epsilon}-m^{5\epsilon /3-1}+m^{\epsilon-1}}
\geq 1-\frac{m^{-\epsilon /3}}{1-o(m)}.
\end{align*}
Choosing $m$ sufficiently large, we get the result
\end{proof}

\begin{proof}[Proof of Theorem \ref{thm:SISpolyLB}]
We assume henceforth that the external agent only tries to infect the (unique) node with maximum degree -- this is without loss of generality for the scaling results. First, note that irrespective of the initial state, the maximum-degree node gets infected with probability at least $\frac{\mu}{1+\mu}=\Omega(1)$ -- thus the expected lifetime with any initial set of infected nodes is at least a constant fraction of the expected lifetime if the epidemic started at the maximum-degree node. Further, note that the lifetime of the epidemic on graph $G$ if some edges are removed is stochastically dominated by the lifetime on $G$ -- this means that to complete the proof, it is sufficient to show that the SIS epidemic has a lifetime which is poly$(n)$ if on a \emph{star graph} with a hub and $n^{\alpha}$ leaf-nodes, and an additional infection-source of rate $\mu=O(1)$ for the hub. We do this by analyzing cyclical epochs characterized by the reinfection/recovery of the hub -- this is similar to techniques used in \cite{massganesh05}, \cite{bbcs05}, however the presence of external infection-sources changes the lifetime dramatically.

We construct the following process which under-dominates the epidemic on the star: at time $0$, we assume that only the hub is infected; subsequently, when \emph{either the hub or any of its neighboring leaf-nodes} recovers for the first time, we instead make the hub recover -- this clearly under-dominates the epidemic, as we have the same recovery rate, but a smaller infection rate in the system. We then show that with probability at least $1-n^{-\Omega(1)}$, the hub gets re-infected before all the infected leaf-nodes get recovered. If this happens, we can again restart the system with only the hub infected (and all leaf-nodes susceptible) -- note that this implies the expected number of such regenerations is $n^{\Omega(1)}$, and the time spent in each epoch is at least $\Omega(1)$, which gives us the desired polynomial expected lifetime.

First, suppose the hub is infected at some time $t$ -- then from Lemma \ref{lem:star}, we have that when any infected node in the neighborhood recovers for the first time after $t$, then the number of infected nodes in the system is $\Omega(n^{\alpha\epsilon/3})$ with probability at least $1-\Omega\left(n^{-\alpha\epsilon/3}\right)$. As we described above, we now assume that the hub is the node which recovers first, and thus there are $\tilde{n}=\Omega(n^{\alpha\epsilon/3})$ infected neighbors.

Next, given that there are $\tilde{n}$ infected neighbors when the hub recovers -- then the probability of the event $A_{\tilde{n}}$ that the hub is not re-infected before all the neighbors recover is given by the following.
$$\PP[A_{\tilde{n}}]=\prod_{i=0}^{\tilde{n}-1}\frac{\tilde{n}-i}{\mu + (\tilde{n} -- i +1) (1+\beta)}$$
To see this, note that if $k$ neighbors are infected, then the probability that one of them recovers before the hub is re-infected is the same as that among three exponential random variables of rate $\mu$, $k\beta$ and $k$, the one with rate $k$ is minimum, i.e., $\frac{k}{\mu + k(1+\beta)}$; subsequently, we can use the memoryless property of the exponential distribution. We can now bound the above as follows:
\begin{align*}
\PP[A_{\tilde{n}}]&=\prod_{i=0}^{\tilde{n}-1}\frac{\tilde{n}-i}{\mu + (\tilde{n} -- i +1) (1+\beta)}\leq \prod_{i=1}^{\tilde{n}} \frac{1}{1+\beta+\frac{\mu}{i}}\\
%&\leq \prod_{i=1}^{\tilde{n}} \exp\left(-\left(\beta+\frac{\mu}{i}\right)+\frac{1}{2}\left(\beta+\frac{\mu}{i}\right)^2\right)\\
&\leq \exp\left(-\sum_{i=1}^{\tilde{n}}\left(\beta+\frac{\mu}{i}\right)+ \sum_{i=1}^{\tilde{n}}\frac{1}{2}\left(\beta+\frac{\mu}{i}\right)^2\right),
\end{align*}
where we use that $1+x\geq e^{x-\frac{x^2}{2}}$. Now, using the fact that $\log n+1\leq  H_n\leq\log(n+1)$, we can simplify to get:
\begin{align*}
&\PP[A_{\tilde{n}}]
%\\ &\leq\exp\left(-\beta\tilde{n} -- \mu(1+\log\tilde{n}) + \frac{\beta^2}{2} \tilde{n} + \beta\mu\log(\tilde{n}+1) + \frac{\pi^2}{12}\mu^2\right)\\
\leq C (\tilde{n}+1)^{-\mu(1-\beta)} = O\left(n^{-\alpha\epsilon\mu/3}\right),
\end{align*}
where we have used the fact that $\beta\lambda_1\leq 1$ and hence is $\beta\leq 1$. Combining the two (via the union bound), we get the probability that the hub gets re-infected before the infection dies in the network is at least $1-O\left(n^{-\alpha\epsilon\mu/3}\right)=1-n^{-\Omega(1)}$. Thus, there are $n^{\Omega(1)}$ such regenerations on average before the infection dies in the network.
\end{proof}

\subsection{Supercritical SIS Epidemics with External Agents}
\label{ssec:SISsupercriticalproofs}

We now give a sufficient condition for supercritical SIS epidemics in the presence of external sources.

\begin{proof}[Proof of Theorem \ref{thm:LBSISgammaI}]
To obtain a sufficient condition for supercritical behavior, we need to lower bound $\mathbb{E}[T_{SIS}]$ -- for this, we stochastically under-dominate the process via the following (one-dimensional) Markov chain: let $Y(t)$ be a continuous time Markov chain on the state space $\{0,1,\cdots,m\}$ with the following transition rates.
\begin{align*}
i \to i+1 &\mbox{ at rate} \ \beta \eta(i) i + \mu \ \mbox{for} \ i>0,\\
i \to i-1 &\mbox{ at rate} \ i \ \mbox{for} \ i>0, 
\end{align*}
and all other transitions having rate $0$. Essentially, $\eta(i)$ captures the bottleneck cut for all sets of size $\leq i$ -- standard coupling arguments show that this process is stochastically dominated by $|\mathbf{X}[t]|$.

To lower bound the time to absorption for this chain we use an approach similar to that in the proof of the theorem \ref{thm:UBSIS} -- we construct an appropriate discrete-time ergodic Markov chain, so that we can use Lemma \ref{lem:ehrenfest}. Note that to find a lower-bound on the expected time to hit zero, it is sufficient to consider an initial state where only $1$ node is infected.

The embedded chain $Y^D$ corresponding to $Y$ has transition probabilities $p_{i,i+1}=\frac{\beta\eta(i)i+\mu}{\beta \eta(i)i+\mu+i}$ and $p_{i,i-1}=\frac{i}{\beta\eta(i)i+\mu+i}$ for all $i>0$ (note we still assume $i\leq m$), and all other transition probabilities being $0$. Defining $T_{1,0}^D$ to be the time for the embedded chain to hit state $0$ starting from $1$, then from Lemma \ref{lem:ehrenfest}), we have:
\begin{equation*}
%\label{eq:supercondn1}
\EE[T_{1,0}^D]=\sum_{k=1}^n\prod_{i=1}^k\left(\beta\eta(i)+\frac{\mu(i)}{i}\right)
\end{equation*}

Next, we want to lower bound the time that the system spends in each state. Let $M=\beta n^2+n+\mu$ -- then the expected time spent at any state of the continuous time Markov chain $Y$ is lower-bounded by $\frac{1}{M}$. Hence $\mathbb{E}[T_{SIS}]=\Omega\left(\frac{e^{n^{\alpha}}}{M}\right)$ which
is $e^{n^{\Omega(1)}}$ as long as $M=O(\mbox{poly}(n))$. 

Thus we have that a sufficient condition for supercritical behavior is given by $\EE[T^D_{1,0}]=\Omega\left(e^{n^{\Omega(1)}}\right)$. Now if $\fall i\in\{1,2,\ldots,n^{\alpha}\}$ for some $\alpha>0$, we have $\beta\eta(i)+\frac{\mu(i)}{i}\geq 1+\epsilon$ for some $\epsilon>0$, then the LHS of the above equation is $\Omega\left(e^{n^\alpha}\right)$, which is what we desire. 

To obtain the condition given in Theorem \ref{thm:LBSISgammaI}, note that random infection strategy is successful at a rate that equals $\tilde{\mu}(i)=\mu(i)\left(1-\frac{i}{n}\right)$, when $|I|=i$. Now suppose we have $\beta\eta(n^{\alpha})+\gamma\geq 1+\epsilon$ for some $\epsilon>0$. Then we have:
\begin{align*}
\EE[T^D_{1,0}]&\geq\sum_{k=1}^{n^{\alpha}}\prod_{i=1}^{k}\left(\beta\eta(n^{\alpha})+\gamma\left(1-\frac{i}{n}\right)\right)\\
&\geq\sum_{k=1}^{n^{\alpha}}\prod_{i=1}^{k}\left(1+\epsilon -\frac{\gamma i}{n}\right). 
\end{align*}
Now we can choose $n$ large enough such that we have $\gamma n^{\alpha-1}<\frac{\epsilon}{2}$ - this then ensures $\EE[T_{1,0}^D]=e^{n^{\Omega(1)}}$.
\end{proof}

\subsection{SIR Epidemics with External Agents}
\label{ssec:SIRproofs}

%As we discussed in Section \ref{ssec:contrib}, studying thresholds for the SIR model with external infection requires understanding the transient behavior of the process, rather than only the steady state behavior. To this end, we utilize the following lemma, which relates the dynamics of the SIR to the SIS model:
%
%\begin{lemma}
%\label{lem:SISubSIRtime}
%For a given graph $G$, contiguous spread rate $\beta$ and external infection policy with virulence $\mu$, $T_{SIS} \geq_{st} T_{SIR}$.
%\end{lemma}
%
%This follows from a simple coupling argument, wherein whenever a node gets recovered in the SIR model, we make it susceptible in the SIS model. 
Finally, we utilize Lemma \ref{lem:SISubSIRtime} to get conditions for subcritical SIR epidemics:

\begin{proof}[Proof of Theorem \ref{thm:SIRsubPoly}]
We first consider a standard technique for bounding the size of eventual infection without external agents (for example, in \cite{draief08}).The main observation is that the probability of a node $i$ to get infected eventually is dependent on its being infected initially and its neighbors $\mathcal{N}(i)=\{j\in V|(i,j)\in E\}$ being infected eventually. We have:
$$\PP[X_i(\infty)=e]\leq\PP[X_i(0)=1]+\frac{\beta}{\beta+1}\sum_{j\in\mathcal{N}(i)}\PP[X_j(\infty)=e]$$

Next, to account for the external sources, we can add an additional term to the RHS to capture the probability that the node $i$ is infected by external agent at some time. Recall that we defined $L_(t)$ to be an inhomogeneous Poisson process representing the external infection at node $i$ -- then the probability of $i$ being infected via the external source is given by $1-\exp(\int_0^{T_{SIR}} L_i(t)dt)$, which is upper-bounded by $\int_0^{T_{SIR}}L_i(t) dt$ for each $i$. Adding up the inequalities for all the nodes, we get:
\begin{align*}
\sum_{i=1}^n\PP[X_i(\infty)=e]&\left(1-\frac{\beta d_i}{\beta+1}\right) \leq\\ &\EE[N_0]+\EE\left[\sum_{i=1}^n\int_0^{T_{SIR}}L_i(t)dt\right],
\end{align*}
where $N_0=\sum_i\mathds{1}_{\{X_i(0)=e\}}$ is the number of initially infected nodes. Reorganizing, and using $d_i\leq d_{\max}$ and $\beta(d_{\max}-1)<1$, we get:
\begin{equation*}
\EE[N_{SIR}] \leq \frac{\EE[N_0]+\mu\EE[T_{SIR}]}{1-\frac{\beta d_{max}}{\beta+1}}
\end{equation*}
To conclude the proof, we need to bound $\EE[T_{SIR}]$ -- however from Lemma \ref{lem:SISubSIRtime}, we know it is sufficient to bound $\EE[T_{SIS}]$. Further, from Theorem \ref{thm:muScales}, we have that if $\mu=o\left(\frac{\log n}{\log\log n}\right)$, then $T_{SIS}=\mbox{subpoly}(n)$. This completes the proof.
\end{proof}

\end{document}